\documentclass[aip,preprint]{revtex4-1}

\usepackage{amssymb,amsmath,amsthm,bm}
\usepackage{amsfonts}
\usepackage{epsfig}
\usepackage{amsbsy}
\usepackage{color}
\usepackage[sort&compress]{natbib}

\usepackage{graphicx}
\graphicspath{{figures/}}

 \newtheorem{lemma}{Lemma}[section]
 \newtheorem{theorem}{Theorem}[section]

\newtheorem{remark}{Remark}[section]

\usepackage{color}
\definecolor{dgreen}{rgb}{0,.6,0}
\usepackage[normalem]{ulem}
\usepackage{bm}

\begin{document}
\preprint{ }

\title{Modulated amplitude waves with nonzero phases in Bose-Einstein
condensates}

\author{Qihuai Liu }
\email{qhuailiu@gmail.com}

\altaffiliation[Corresponding author at: ]{School of Mathematics
and Computing Science, Guilin University of Electronic Technology,
No. 1, Jinji Street, Guilin 541004, China. Tel./Fax: +086 0773
3939803. }

\affiliation{School of Mathematics and Computing Science, Guilin
University of Electronic Technology, Guilin 541002, P. R. China}

\author{Dingbian Qian }
\affiliation{School  of Mathematical Sciences, Soochow University,
Suzhou 215006, P. R. China}

\date{\today}

\begin{abstract}
In this paper we give a frame for application of the averaging
method to Bose-Einstein condensates (BECs) and obtain an abstract
result upon the dynamics of BECs. Using the averaging method, we
determine the location where the modulated amplitude waves
(periodic or quasi-periodic) exist and obtain that all these
modulated amplitude waves (periodic or quasi-periodic) form a
foliation by varying the integration constant continuously.
Compared with the previous work, modulated amplitude waves studied
in this paper have nontrivial phases and this makes the problem
become more difficult, since it involves some singularities.

\end{abstract}



\maketitle

\section{Introduction}
\label{Sec 1} The experimental realization of Bose-Einstein
condensates (BECs) in dilute alkali-metal atomic vapors
\cite{anderson1995observation,davis1995bose} has sparked a large
mathematical and physical interest in the study of dynamics of
condensates, such as solitons
\cite{strecker2002formation,khaykovich2002formation,burger1999dark,denschlag1999guiding,belmonte2008existence,belmonte2009solitary},
chaos \cite{chong2004spatial,xie2003chaotic}, stability and
instability~\cite{centurion2006modulational,kapitula1998stabilitya,kapitula2010interaction,jones1993stability,kapitula1998stabilityb},
periodic and quasi-periodic behaviors
\cite{van2007quasiperiodic,deconinck2002dynamics}.

At ultra-low temperatures, based on mean-field approximation and
quasi-one-dimensional [quasi-(1D)] regime, the time-dependent
condensate wave function (``order parameter'') $\psi(x,t)$ is
governed by the following cubic nonlinear schr\"{o}dinger equation
(NLS) \cite{dalfovo1999theory,baizakov2002regular,khler2002three}
\begin{equation}\label{1}
 i\hbar\frac{\partial\psi}{\partial t}=-\frac{\hbar^2}{2m}\frac{\partial^2\psi}{\partial
 x^2}+g|\psi|^2\psi+V_0(x)\psi,
\end{equation}
which is also known as the Gross-Pitaevskii (GP) equation. Here,
$|\psi|^2$ is the number density, $V_0(x)$ is an external
potential, $g=[4\pi\hbar^2 a/m][1+\mathcal{O}(\zeta^2)]$, and
$\zeta=\sqrt{|\psi|^2|a|^3}$ is the dilute gas parameter. The
$s$-wave scattering length $a$ is determined by the atomic species
of the condensate. Interactions between atoms are repulsive when
$a>0$ and attractive when $a<0$. In collisionally inhomogeneous
BECs, the scattering length is subjected to a spatial periodic
variation: $a(x)=a(x+L_0)$ for some period $L_0$ leading to
nonlinear potentials, which has been realized experimentally
\cite{carpentier2006analysis,abdullaev2005propagation} and studied
theoretically \cite{porter2007modulated,tsang2010exact}. Also, we
can refer to a recent comprehensive review
\cite{kartashov2011solitons}.

Spatially periodic potentials $V(x)$, created as optical lattices
(OLs), which arise as interference patterns produced by coherent
counterpropagating laser beams illuminating the condensate, are of
interest in the context of BECs and have been employed in both
experimental and theoretical studies
\cite{anderson1998macroscopic,hagley1999well,bronski2001bose,choi1999bose,trombettoni2001discrete}.

In this study, we investigate spatially extended solutions of BECs
in periodic OLs. We apply a coherent structure ansatz to
(\ref{1}), yielding a parametrically forced Duffing equation with
singularities, which describes the spatial evolution of the field.
We employ the averaging method to study the periodic orbits (MAWs)
including their hyperbolic spatial structures, and illustrate
their dynamical behaviors with numerical simulation of the GP
equation.

Compared with the previous work
\cite{porter2004modulated,porter2007modulated,porter2004perturbative,porter2006feshbach,porter2004resonant,chua2006spatial,porter2005bose},
the phases of MAWs considered in this paper are nontrivial, not
constant, which continuously depend on the spatial variation.
Nontrivial phase solutions have more complicated dynamics
\cite{chong2004spatial,hai2004propagation} and imply nonzero
current of the matter - it is proportional to $R^2(x)\theta(x)=c$,
for amplitude $R(x)$ of MAW and nonzero constant $c$ -along
$x$-axis, and hence seem to have no direct relation to present
experimental setting for BECs
\cite{kostov2004two,carr2001stability}(remember that the
condensate in this paper is confined to be a parabolic trap). Of
course, dealing with the MAWs with nonzero phases will take more
difficulties than before. One reason is that the forced Duffing
oscillator derived from GP equation has singularities at the
origin and the directly application of the usual perturbation
theory is unavailable because of the loss of smoothness and the
existence of strong nonlinear term (singular term).

The averaging method
\cite{bogoliubovasymptotic,sanders2007averaging,guckenheimer1983nonlinear}
at its heart is a transformation procedure leading to a systematic
perturbation expansion and completes with error bounds on the
difference between exact and approximation solutions. Also, it is
a important tool for proving properties of the exact problem based
on properties of the approximation problem
\cite{ellison1995method}. For example, the existence of periodic
orbits can be proved using averaging together with the implicit
function theorem, and the existence of invariant tori can be
proved using averaging together with the Moser twist theorem.
However, applying the averaging method for singular system, the
first problem we must deal with is how to put the system into the
standard form of averaging.

The paper is organized as follows. In Section \ref{SEC 2}, we
introduce modulated amplitude waves involving periodic and
quasi-periodic, and in Section \ref{Sec 3}, we apply a
transformation to transform the GP equation to a standard form of
averaging. An abstract result upon the dynamics to BECs  is
obtained in Section \ref{Sec 4}, and in Section \ref{Sec 5} we
analyze and demonstrate some of the spatial dynamical features of
BECs with a positive chemical potential. Finally, we summarize our
results in Section \ref{Sec 6}.

\section{Coherent structure and modulated amplitude wave}\label{SEC 2}\setcounter{equation}{0}

We consider uniformly propagating coherent structures with the
ansatz
\begin{equation}\label{2}
\psi(t,x)=R(x)\exp(i[\Theta(x)-\mu t]),
\end{equation}
where $R(x)\in\mathbb{R}$ gives the amplitude dynamics of the
condensate wave function, $\theta(x)$ determines the phase
dynamics, and the ``chemical potential'' $\mu$, defined as the
energy which takes to add one more particle to the system, is
proportional to the number of atoms trapped in the condensate.
When the (temporally periodic) coherent structure (\ref{2}) is
also spatially periodic, it is called a \emph{modulated amplitude
wave} (MAW) \cite{brusch2000modulated,brusch2001modulated}.
Similarly, a solution of the equation (\ref{1}) with the
(temporally periodic) coherent structure (\ref{2}) is called a
\emph{quasi-periodic modulated amplitude wave} (QMAW) if it is
also spatially quasi-periodic.

Inserting (\ref{2}) into (\ref{1}), we obtain the following two
couple nonlinear ordinary differential equations
\begin{equation}\label{3}
R''+\delta R-\frac{c^2}{R^3}+\varepsilon\alpha R^3+\varepsilon
V(x)R=0,
\end{equation}
\begin{equation}\label{4}
\Theta''+2\Theta'R'/R=0~\Rightarrow~\Theta'(x)=\frac{c}{R^2},\quad\quad
\end{equation}
where
\begin{equation*}
    \delta:=\frac{2m\mu}{\hbar},~\varepsilon\alpha:=-\frac{2mg}{\hbar^2},~\varepsilon
    V(x):=-\frac{2m}{\hbar^2}V_0(x)
\end{equation*}
and the integration constant $c$, determined by the velocity and
number density, plays the role of ``angular momentum''
\cite{bronski2001bose}.

 Inspecting (\ref{3}) we know that in case of $c=0$, i.e., the
 phase of the condensate wave function is trivial, it is the
 parametrically driven Duffing equation with the time variable
 replaced by the spatial coordinate, and MAWs (standing waves) in this
 system with $V(x)=V_0\cos\kappa x$ or $V(x)=V_1\cos\kappa_1 x+V_2\cos\kappa_2x$ have been widely studied \cite{porter2004perturbative,porter2007modulated,porter2005bose}.

 In general, $c\neq 0$, the system (\ref{3}) becomes more
 complicated and the phase is no longer constant
 \cite{chong2004spatial}. Even the amplitude $R(x)$, a solution of
 (\ref{3}), is $L$-periodic, the corresponding condensate wave
 function $\psi(x,t)$ may be not periodic with respect to the
 spatial variable $x$. In fact,
 \begin{eqnarray*}
    \psi(t,x)&=&R(x)\mathrm{exp}{i[\Theta(x)-\mu t]}\\
    &=&R(x)\big(\cos[\bar{\Theta}(x)+\nu x-\mu t]+i\sin[\bar{\Theta}(x)+\nu x-\mu
    t]\big),
\end{eqnarray*}
where
\begin{equation*}
    \nu=\frac{1}{L}\int_{x_0}^{x_0+L}\frac{c}{R^2(\xi)}\mathrm{d}\xi
\end{equation*}
and $\bar{\Theta}(x)=\Theta(x)-\nu$ is a $L$-periodic function
with zero mean value. If $2\pi/\nu$ and $L$ are rationally
related, then $\psi(x,t)$ is a MAW; if $2\pi/\nu$ and $L$ are
rationally irrelevant, then $\psi(x,t)$ is not periodic but
quasi-periodic, which is corresponding to a QMAW with the
frequency $\omega=\langle2\pi/\nu, L\rangle$.

There also exists an interesting and surprising result. Note that
$\nu=0$ when $c=0$ and $\nu\neq0$ when $c\neq0$. If we vary $c$ on
the interval $(-\infty,+\infty)$, by continuous dependence of
solutions with respect to the parameters, $\nu$ can continuously
take the value on some interval, which implies that (\ref{1}) has
infinitely many (positive measure set) MAWs and QMAWs by adjusting
the integration constant $c$. Thus, all these MAWs and QMAWs form
a foliation.

In this paper, we consider the case $\delta>0$ corresponding to a
positive chemical potential. Also, in order that the mathematical
results obtained in this paper do apply to more general periodic
functions, we assume that the external potential $V(x)$ is an
analytic and $L$-periodic function (OLs). Note that (\ref{3})
defines on two half-planes, and we only consider the case of the
right half-plane since there are no distinct technicalities.

\section{Transformation to standard form of averaging}\label{Sec
3}\setcounter{equation}{0}

Rewrite equation (\ref{3}) in the planar equivalent form
\begin{equation}\label{3.1}
\left\{ \begin{array}{ll}
R'=S\\
S'=-\delta R+\displaystyle\frac{c^2}{R^3}-\varepsilon \alpha
R^3-\varepsilon V(x)R.
\end{array}
\right.
\end{equation}
Generally, averaging method involves two steps: transforming to
standard form; solving the averaging equation. In order to proceed
we need to transform (\ref{3.1}) to a standard form for the method
of averaging. So we have the following result.

\begin{lemma}\label{LM 3.1} Under the transformation $\Psi:\mathbb{T}\times
\big(\sqrt[4]{\frac{c^2}{\delta}},+\infty \big)\rightarrow
(0,+\infty)\times \mathbb{R}$ defined by
\begin{equation*}
\left\{ \begin{array}{ll} R
=\rho\sqrt{\cos^2(\sqrt{\delta}x+\theta)+\displaystyle\frac{c^2}{\delta\rho^4}\sin^2(\sqrt{\delta}x+\theta)}\\[2em]
S
=\rho\sqrt{\delta}\left(\displaystyle\frac{c^2}{\delta\rho^4}-1\right)\displaystyle\frac{\cos(\sqrt{\delta}x+\theta)\sin(\sqrt{\delta}x+\theta)}
{\sqrt{\cos^2(\sqrt{\delta}x+\theta)+\displaystyle\frac{c^2}{\delta\rho^4}\sin^2(\sqrt{\delta}x+\theta)}}
,
\end{array}
\right.
\end{equation*}
system \emph{(\ref{3.1})} changes into a new system
\begin{equation}\label{3.2}
\left\{ \begin{array}{llll} \rho'
=\varepsilon\left\{\displaystyle\frac{\alpha}{\sqrt{\delta}}\rho^3\left[\frac{1}{4}(1+\frac{c^2}{\delta\rho^4})\sin2(\sqrt{\delta}x+\theta)
+\frac{1}{8}(1-\frac{c^2}{\delta\rho^4})\sin4(\sqrt{\delta}x+\theta)\right]\right.\\
~~~~~~~~~~~~~~~~~~~~~~~~~~~~~~~~~~~~~+\left.\displaystyle\frac{\rho}{2\sqrt{\delta}}V(x)\sin2(\sqrt{\delta}x+\theta)\right\}\\[2em]
\theta'
=\varepsilon\left\{\displaystyle\frac{\alpha(\delta\rho^4+c^2)}{8\delta^2\rho^2}\big(3+\cos4(\sqrt{\delta}x+\theta)\big)+
\displaystyle\frac{\alpha(\delta^2\rho^8+c^4)}{2\delta^2\rho^2(\delta\rho^4-c^2)}\cos2(\sqrt{\delta}x+\theta) \right.\\[1em]
~~~~~~~~~~~~~~~~~~~~~~~~~~~~~~~~~~~~~+\left.\displaystyle\frac{1}{2\delta}V(x)\left(1+\frac{\delta\rho^4+c^2}{\delta\rho^4-c^2}\cos2(\sqrt{\delta}x+\theta)\right)\right\}
\end{array}
\right.
\end{equation}
with the new coordinates $(\theta,\rho)$ in the half-plane
$\mathbb{T}\times \big(\sqrt[4]{\frac{c^2}{\delta}},+\infty
\big)$.
\end{lemma}

\noindent\textbf{Proof.} First, it is easy to verify that, for
each $\rho\in\big(\sqrt[4]{\frac{c^2}{\delta}},+\infty \big)$ and
$\theta\in \mathbb{T}$,
\begin{equation*}
    \varphi(x;\theta,\rho)=\rho\sqrt{\cos^2(\sqrt{\delta}x+\theta)+\displaystyle\frac{c^2}{\delta\rho^4}\sin^2(\sqrt{\delta}x+\theta)}
\end{equation*}
is a periodic solution with the same period
$\tau=2\pi/\sqrt{\delta}$ of the unperturbed equation
\begin{equation*}
R''+\delta R-\frac{c^2}{R^3}=0
\end{equation*}
or the equivalent Hamiltonian system
\begin{equation}\label{3.3}
\left\{ \begin{array}{ll} R' =S\\ S' =-\delta
R+\displaystyle\frac{c^2}{R^3}
\end{array}
\right.
\end{equation}
with the initial value
\begin{equation*}
\varphi(-\frac{\theta}{\sqrt{\delta}};\theta,\rho)=\rho,~~\varphi'(-\frac{\theta}{\sqrt{\delta}};\theta,\rho)=0.
\end{equation*}
That is to say, (\ref{3.3}) is an isochronous system. The total
energy of system (\ref{3.3}) is given by
\begin{equation*}
    \frac{1}{2}S^2+V(R)=V(\rho),
\end{equation*}
where
\begin{equation*}
    V(R)=\frac{\delta R^2}{2}+\frac{c^2}{2R^2},
\end{equation*}
and as we know, all periodic solutions lie on curves of constant
energy. We will use these facts below.

Now using the variation of constants, the functions $\rho(\cdot),
\theta(\cdot)$ can be defined such that
\begin{equation*}
    R(x)=\varphi(x+\theta(x),\rho(x)),~~S(x)=\frac{\partial\varphi}{\partial
    x}(x+\theta(x),\rho(x)).
\end{equation*}
From the conservation of the Hamiltonian, it follows that
\begin{equation}\label{3.4}
 \frac{1}{2}\Big(\frac{\partial\varphi}{\partial x}\Big)^2+V(\varphi)=V(\rho),
\end{equation}
so the differentiating with respect to $x$ along with system
(\ref{3.1}) yields
\begin{align*}
     V'(\rho)\displaystyle\frac{\mathrm{d}\rho}{\mathrm{d}x}&=\left.\displaystyle\frac{\mathrm{d}~
     }{\mathrm{d}x}\left(\frac{1}{2}S^2(x,\theta,\rho)+V(R(x,\theta,\rho))\right)\right|_{(\ref{3.1})}\\
     &=-\varepsilon\left[\alpha
     R^3(x,\theta,\rho)+V(x)R(x,\theta,\rho)\right]S(x,\theta,\rho).
\end{align*}
Then, we have
\begin{equation}\label{3.5}
\frac{\mathrm{d}\rho}{\mathrm{d}x}=-\varepsilon\frac{1}{V'(\rho)}\left[\alpha
     R^3(x,\theta,\rho)+V(x)R(x,\theta,\rho)\right]S(x,\theta,\rho).
\end{equation}
So, using the definition of $\Psi$ together with equation
(\ref{3.5}), we have the first desired expression of (\ref{3.2}).

Using the formula for $R'$ given in system (\ref{3.1}) and the
definition of $\rho(\cdot)$ and $\theta(\cdot)$, we obtain that
\begin{equation*}
\frac{\partial \varphi}{\partial
x}\left(1+\frac{\mathrm{d}\theta}{\mathrm{d}x}\right)+\frac{\partial
\varphi}{\partial
\rho}\frac{\mathrm{d}\rho}{\mathrm{d}x}=\frac{\partial
\varphi}{\partial x}.
\end{equation*}
Together with equation (\ref{3.5}), after some simple algebraic
manipulations, it follows that
\begin{equation*}
\frac{\mathrm{d}\theta}{\mathrm{d}x}=\varepsilon\frac{1}{V'(\rho)}\left[\alpha
     R^3(x,\theta,\rho)+V(x)R(x,\theta,\rho)\right]\frac{\partial
\varphi}{\partial \rho}.
\end{equation*}
Finally, using the definition of $\Psi$ together with equation
(\ref{3.5}), we have the second desired expression of
(\ref{3.2}).\qed

In general, finding a explicit expression of the transformation is
not easy. The transformation $\Psi$ given in Lemma \ref{LM 3.1}
has a more delicate information, such as protecting two-form
\begin{equation*}
    \mathrm{d}R\wedge
    \mathrm{d}S=\sqrt{\delta}\big(\rho-\frac{c^2}{\delta\rho^3}\big)\mathrm{d}\rho\wedge\mathrm{d}\theta,
\end{equation*}
which is not used in this paper. However, we believe that it will
be helpful for further study.

\section{An abstract result of averaging to BECs}\label{Sec
4}\setcounter{equation}{0}

In this section, we will give an abstract result to BECs by the
method of averaging. We also assume that the period $L$ of the
external potential $V(x)$ satisfies that
\begin{equation*}
    \frac{L}{\tau}=\frac{q}{p}\in\mathbb{Q},~~(q,p)=1,~p,q\in\mathbb{Z}^+,
\end{equation*}
where $\tau$ is the least period of the unperturbed system. The
basic idea that leads to the application of the method of
averaging arises from an inspection of system (\ref{3.2}). The
derivatives with respect to the spatial variable $x$ are all
proportional to $\varepsilon$. Hence, if $\varepsilon$ is small,
the variables would be expected to remain near their constant
unperturbed values over a long spatial scale.

A good approximation of (\ref{3.2}) up to the spatial domains of
order $1/\varepsilon$ is given by the averaged system
\begin{equation}\label{4.1}
\left\{ \begin{array}{ll}
\bar{\rho}'=\varepsilon\bar{\rho}\Phi(\bar{\theta})\\[0.5em]
\bar{\theta}'=\varepsilon\Big(\displaystyle\frac{\alpha_0}
{\bar{\rho}^2}(\delta\bar{\rho}^4+c^2)+\alpha_1+\displaystyle\frac{\delta\bar{\rho}^4+c^2}{2\sqrt{\delta}(\delta\bar{\rho}^4-c^2)}\Phi'(\bar{\theta})\Big),
\end{array}
\right.
\end{equation}
where
\begin{align*}
\Phi(\bar{\theta})&=\frac{1}{2\sqrt{\delta}\tilde{L}}\int_0^{\tilde{L}}V(s)\sin2(\sqrt{\delta}s+\bar{\theta})\mathrm{d}s,\\
\alpha_0&=\frac{3\alpha}{8\delta^2\tilde{L}},\\
\alpha_1&=\frac{1}{2\delta
\tilde{L}}\int_0^{\tilde{L}}V(s)\mathrm{d}s,~~\tilde{L}=qp\min\{L,\tau\}.
\end{align*}
We have the following theorem.
\begin{theorem}\label{TH 4.1}
There exists a $c^r,r\geq2$, change of variables
\begin{equation*}
    \rho=\bar{\rho}+\varepsilon
    w_1(\bar{\theta},\bar{\rho},x,\varepsilon),~~  \theta=\bar{\theta}+\varepsilon w_2(\bar{\theta},\bar{\rho},x,\varepsilon)
\end{equation*}
with $w_1,w_2$ $\tilde{L}$-periodic functions of $x$, transforming
\emph{(\ref{3.2})} into
\begin{equation*}
\left\{ \begin{array}{ll}
\bar{\rho}'=\varepsilon\bar{\rho}\Phi(\bar{\theta})+\varepsilon^2
    g_1(\bar{\theta},\bar{\rho},x,\varepsilon)\\[0.5em]
\bar{\theta}'=\varepsilon\Big(\displaystyle\frac{\alpha_0}
{\bar{\rho}^2}(\delta\bar{\rho}^4+c^2)+\alpha_1+\displaystyle\frac{\delta\bar{\rho}^4+c^2}{2\sqrt{\delta}(\delta\bar{\rho}^4-c^2)}\Phi'(\bar{\theta})\Big)
+\varepsilon^2
    g_2(\bar{\theta},\bar{\rho},x,\varepsilon)
\end{array}
\right.
\end{equation*}
with $g_1,g_2$ $\tilde{L}$-periodic functions of $x$. Moreover,

\emph{(i)} If $(\theta_\varepsilon(x),\rho_\varepsilon(x))$ and
$(\theta_0(x),\rho_0(x))$ are solutions of the original system
\emph{(\ref{3.2})} and the averaged system \emph{(\ref{4.1})}
respectively, with the initial values such that
\begin{equation*}
    |\,\rho_\varepsilon(0)-\rho_0(0)|+|\,\theta_\varepsilon(0)-\theta_0(0)|=\mathcal{O}(\varepsilon),
\end{equation*}
then
\begin{equation*}
    |\,\rho_\varepsilon(x)-\rho_0(x)|+|\,\theta_\varepsilon(x)-\theta_0(x)|=\mathcal{O}(\varepsilon),
\end{equation*}
for the spatial domains $x$ of order $1/\varepsilon$.

\emph{(ii)} If there exist two constants
$\rho_0\in(\sqrt[4]{\frac{c^2}{\delta}},+\infty \big),
\theta_0\in\mathbb{R}$ such that
\begin{align}
 \label{4.2}   &\Phi(\theta_0)=0,\\
\label{4.3}  &\displaystyle\frac{\alpha_0}
{\bar{\rho}_0^2}(\delta\bar{\rho}_0^4+c^2)+\alpha_1+\displaystyle\frac{\delta\bar{\rho}^4_0+c^2}{2\sqrt{\delta}(\delta\bar{\rho}^4_0-c^2)}\Phi'(\bar{\theta}_0)
=0,\\
\label{4.4} &\Phi'(\theta_0)\neq0,\\
\label{4.5}
&2\alpha_0(\delta\rho_0-\frac{c^2}{\rho^2_0})-\frac{4\sqrt{\delta}c^2\rho_0^3}{(\delta\bar{\rho}^4_0-c^2)^2}\Phi'(\theta_0)\neq
0,
\end{align}
i.e., $(\theta_0,\rho_0)$ is an equilibrium point of
\emph{(\ref{4.1})} such that the corresponding Jacobian matrix has
no eigenvalue equaling to zero, then \emph{(\ref{3.2})} admits a
$\tilde{L}$-periodic solution
$(\theta_\varepsilon(x),\rho_\varepsilon(x))$ such that
\begin{equation*}
    |\,\rho_\varepsilon(x)-\rho_0|+|\,\theta_\varepsilon(x)-\theta_0|=\mathcal{O}(\varepsilon),
\end{equation*}
for sufficiently small $\varepsilon$; if, in addition,
\begin{align}
\label{4.6}
2\alpha_0(\delta\rho_0-\frac{c^2}{\rho^2_0})\Phi'(\theta_0)-\frac{4\sqrt{\delta}c^2\rho_0^3}{(\delta\bar{\rho}^4_0-c^2)^2}[\Phi'(\theta_0)]^2>
0,
\end{align}
then the $\tilde{L}$-periodic solution
$(\rho_\varepsilon(x),\theta_\varepsilon(x))$ is hyperbolic and
instable with respect to the spatial variable $x$.
\end{theorem}

\begin{proof}
Conditions (\ref{4.2})-(\ref{4.6}) imply that $(\theta_0,\rho_0)$
is a instable and hyperbolic fixed point of system (\ref{4.1}).
The proof of this theorem follows directly from \cite[Theorem
3.2.3]{berglund2001perturbation} or \cite[Theorem
4.1.1]{guckenheimer1983nonlinear}.
\end{proof}

\begin{remark}\label{RM 4.1} The instability in
Theorem 4.1 is only relevant to the amplitude equation (2.2),
which is some artificial ``instability" in terms of the evolution
in $x$. There is a set of methods for the study of modulational
instability in time $t$, e.g., see
\cite{zp2xxx,zg2004,ha1999,zb1998}. Recently, based on spectrum
theory and Hamiltonian floquet theory, the method of studying
modulational (temporal) instability of standing waves (with
trivial phases) has been developed for NLS equation with constant
nonlinearity coefficients and periodic potentials by Bronski and
Rapti\cite{bronski2005modulational}, later applied by Porter et
al. \cite{porter2007modulated}. However, here we can not provide
any information upon it with our methods.
\end{remark}

\section{Equilibriums and the averaged equation}\label{Sec
5}\setcounter{equation}{0}

In this section, we will analyze and demonstrate some of the
spatial dynamical features of BECs with a positive chemical
potential. We also assume that $V(x)$ is an analytic and
$L$-periodic function with the least positive period
$L=\tau=\pi/\sqrt{\delta}$, i.e., $p=q=1$.

First, expanding $V$ in a Fourier series, we have
\begin{equation}\label{5.1}
V(x)=b_0+\sum_{k=1}^\infty[
a_k\sin(2\sqrt{\delta}kx)+b_k\cos(2\sqrt{\delta}kx)],
\end{equation}
where all coefficients are real. Let us substitute the expansion
for $V(x)$ given by (\ref{5.1}) into (\ref{4.1}), and by an easy
(perhaps lengthy) computation, we obtain the averaging system
\begin{equation}\label{5.2}
\left\{ \begin{array}{lll} \bar{\rho}'
=\varepsilon\displaystyle\frac{\bar{\rho}}{4\sqrt{\delta}}\sqrt{a_1^2+b_1^2}\sin(2\bar{\theta}+\phi)\\[1em]
\bar{\theta}'
=\displaystyle\frac{3\varepsilon\alpha(\delta\bar{\rho}^4+c^2)}{8\delta^2\bar{\rho}^2}+\displaystyle\frac{\varepsilon
b_0}{2\delta}+\displaystyle\frac{\varepsilon
}{4\delta}\displaystyle\frac{\delta\bar{\rho}^4+c^2}{\delta\bar{\rho}^4-c^2}\sqrt{a_1^2+b_1^2}\cos(2\bar{\theta}+\phi),
\end{array}
\right.
\end{equation}
where
\begin{align*}
    \phi&=\displaystyle\arctan\frac{a_1}{b_1}
    ~~~(\phi=\displaystyle\mathrm{Sign}(a_1)\cdot\pi/2,~\text{if}~b_1=0).
\end{align*}

If ${\alpha}<0$, corresponding to the repulsive nonlinearity, we
take $\theta_0=k\pi-\phi/2$. Let us define the function
$f:\big(\sqrt[4]{\frac{c^2}{\delta}},+\infty \big)\rightarrow
\mathbb{R}$ by
\begin{equation}\label{5.3}
   f(\rho)=\displaystyle\frac{3\alpha(\delta\bar{\rho}^4+c^2)}{8\delta^2\bar{\rho}^2}+\displaystyle\frac{\varepsilon
b_0}{2\delta}+\displaystyle\frac{\varepsilon
}{4\delta}\displaystyle\frac{\delta\bar{\rho}^4+c^2}{\delta\bar{\rho}^4-c^2}\sqrt{a_1^2+b_1^2}.
\end{equation}
It is easy to verify that
\begin{equation*}
f(\rho)\rightarrow-\infty,~\text{as}~\rho\rightarrow+\infty;~f(\rho)\rightarrow+\infty,~\text{as}~\rho\rightarrow\sqrt[4]{c^2/\delta}^+.
\end{equation*}
Using the mean-value theorem, there is at least one root $\rho_0$
of $f$ on the interval $\big(\sqrt[4]{\frac{c^2}{\delta}},+\infty
\big)$. If ${\alpha}>0$, corresponding to the attractive case, we
take $\theta_0=k\pi+(\pi-\phi)/2$. Similarly, we can induce that
the function $f$ also has at least one root $\rho_0$ on the
interval $\big(\sqrt[4]{\frac{c^2}{\delta}},+\infty \big)$. Thus,
in any case, the averaged system (\ref{5.2}) has at least one
equilibrium $(\theta_0,\rho_0)$.

After a simple computation, the Jacobian matrix of the averaged
system (\ref{5.2}) at the equilibrium $(\theta_0,\rho_0)$ is given
by
\begin{equation}\label{5.4}
M=\left(%
\begin{array}{cc}
  0~~~~~~~~~~~~~~~~~~~~~~~~~~~~~~~~~~~~\displaystyle\frac{\varepsilon\rho_0\sqrt{a_1^2+b_1^2}\cos(2\theta_0+\phi)}{2\sqrt{\delta}}
  \\[1em]
  \varepsilon\displaystyle\frac{3\alpha(\delta\rho_0^4+c^2)[\frac{1}{4\delta}(\delta\rho_0^4-c^2)^2+c^2\rho_0^4]+4c^2b_0\delta\rho_0^6}
  {\delta\rho_0^3(\delta\rho_0^4+c^2)(\delta\rho_0^4-c^2)}
  ~~~~~~~~~~~0  \\
\end{array}%
\right).
\end{equation}
Notice that if $M$ has no zero eigenvalue, then the equilibrium
$(\theta_0,\rho_0)$ can be continuable.  If $M$ has no imaginary
eigenvalue, the equilibrium $(\theta_0,\rho_0)$ is hyperbolic
since the two eigenvalues of $M$ have opposite signs. So it
follows that the original system (\ref{3.2}) has a hyperbolic
$L$-periodic solution
$(\theta_\varepsilon(x),\rho_\varepsilon(x))$ with respect to the
spatial variable $x$ near $(\theta_0,\rho_0)$ such that
$(\theta_\varepsilon(x),\rho_\varepsilon(x))\rightarrow
(\theta_0,\rho_0)$, as $\varepsilon\rightarrow0$.

To demonstrate the process of averaging to BECs, a specific
example of numerical computation is given in the following. We
take the integration constant $c=1$ and the parameters
$\delta=1,\alpha=0.3,b_0=-12.4,b_1=-20.19,a_1=0$. By a numerical
computation via MATHMETICA, we obtain equilibriums
\begin{equation*}
(k\pi+\frac{\pi}{2},3.00),~(k\pi+\frac{\pi}{2},2.00),~(k\pi,10.00),k=0,\pm1,\pm2,\cdots
\end{equation*}
for the averaged system (\ref{5.2}).

\begin{figure}[htbp]
  \begin{center}
      \includegraphics[scale=0.8]{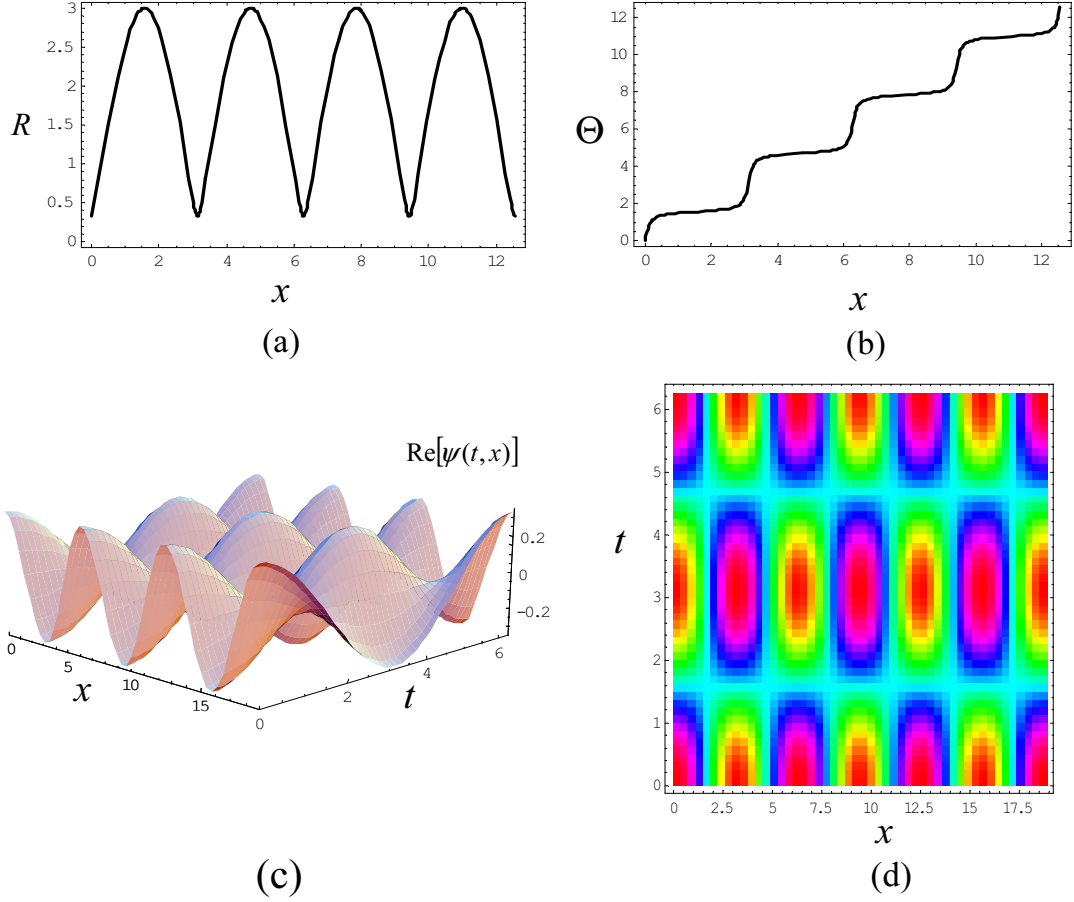}
  \end{center}
\caption{A plot of a MAW $\psi_0(x,t)$ corresponding to the
equilibrium $(\frac{\pi}{2},3)$ for unperturbed system of
(\ref{1}) on the right-plane. The original system (\ref{1}) has a
  MAW or QMAW $\psi_\varepsilon(x,t)$ such
that
$\big|\,|\psi_0(x,t)|-|\psi_\varepsilon(x,t)|\big|=\mathcal{O}(\varepsilon)$
 for all time $t$ and the spatial domains of order $1/\varepsilon$. The parameters we take
 as follows: $c=1,~\delta=1,~\alpha=0.30,~b_0=-12.40,~b_1=-20.20$. (a)~Spatial amplitude $R(x)$ plot; (b)~Nontrivial spatial phase $\Theta(x)$~plot;
(c) A plot of space-time $\mathrm{R}e[\psi(t,x)]$ ; (d) The
density plot of $\mathrm{Re}[\psi(t,x)]$.} \label{Fig 5.1}
\end{figure}

\begin{figure}[htbp]
  \begin{center}
      \includegraphics[scale=0.8]{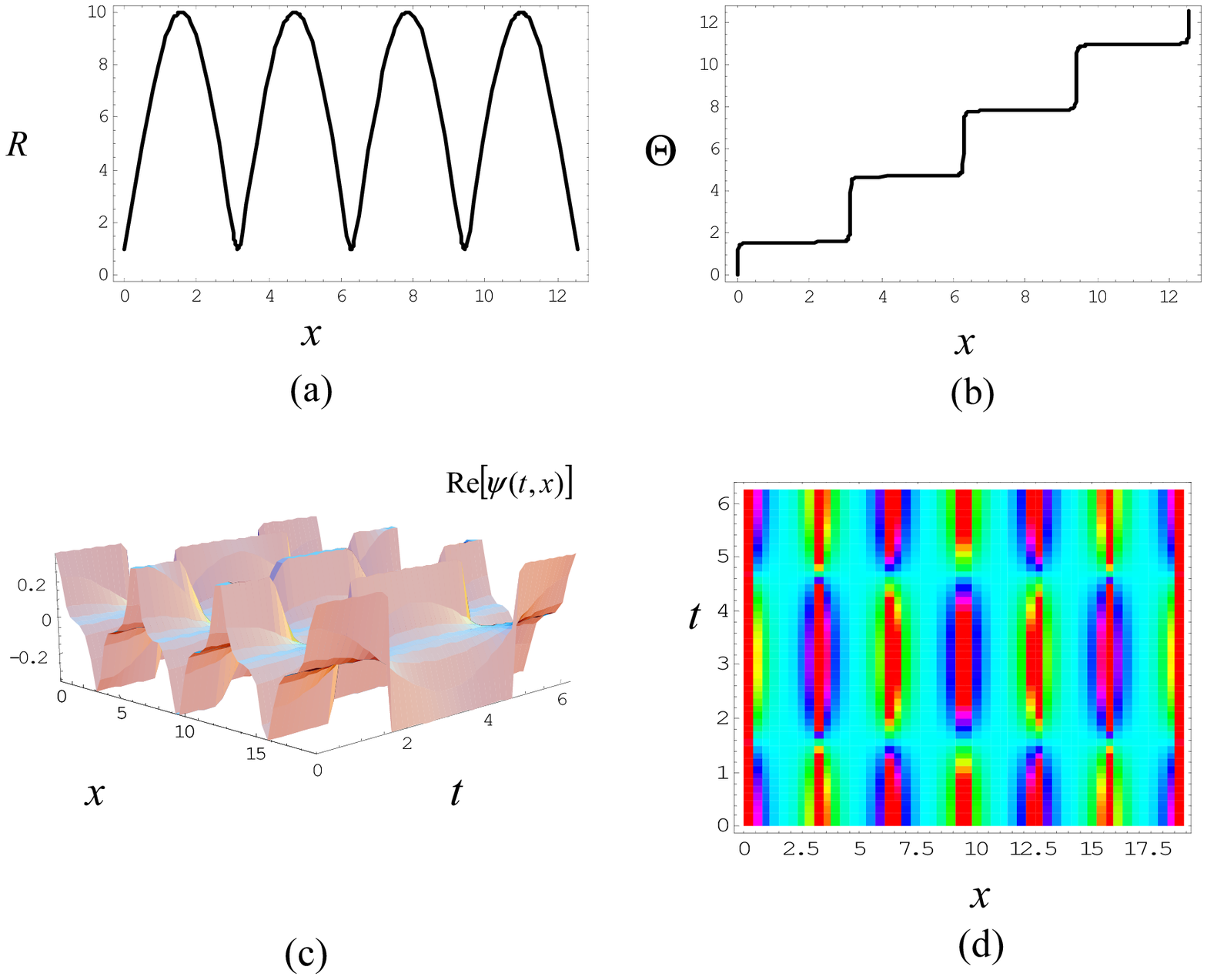}
  \end{center}
\caption{A plot of a MAW $\psi_0(x,t)$ corresponding to the
equilibrium $({\pi},10)$ for unperturbed system of (\ref{1}) on
the right-plane. The original system (\ref{1}) has a  MAW or QMAW
$\psi_\varepsilon(x,t)$ such that
 $\big|\,|\psi_0(x,t)|-|\psi_\varepsilon(x,t)|\big|=\mathcal{O}(\varepsilon)$
 for all time $t$ and the spatial domains of order $1/\varepsilon$. The parameters we take
 as follows: $c=1,~\delta=1,~\alpha=0.30,~b_0=-12.40,~b_1=-20.20$. (a)~Spatial amplitude $R(x)$ plot; (b)~Nontrivial spatial phase $\Theta(x)$~plot;
(c) A plot of space-time $\mathrm{R}e[\psi(t,x)]$ ; (d) The
density plot of $\mathrm{Re}[\psi(t,x)]$.} \label{Fig 5.2}
\end{figure}

 As we know, $(\theta_0,\rho_0)=(k\pi+\pi/2,3.00)$ and
$(\theta_1,\rho_1)=(k\pi,10.00)$ are hyperbolic with the
eigenvalue of linearization
\begin{equation*}
    \lambda_0=\pm0.95\varepsilon,~\lambda_1=\pm15.33\varepsilon.
\end{equation*}
By the averaging theorem, the equilibriums can persist as the
periodic orbits for the original system (\ref{3.1}); in addition,
these periodic orbits are also hyperbolic with respect to spatial
variable $x$.

Returning to (\ref{3.1}), consider its unperturbed system, and
using Lemma \ref{LM 3.1}, the periodic orbit corresponding to the
equilibrium $(\theta_0,\rho_0)=(k\pi+\pi/2,3.00)$ is given by
\begin{equation*}
    \varphi_0(x)=3\sqrt{\sin^2x+\frac{1}{81}\cos^2x}.
\end{equation*}
Following from (\ref{4}), we have the angle function
\begin{equation*}
    \Theta_0(x)=\int_0^x\frac{1}{{\varphi_0}^2(s)}\mathrm{d}s
\end{equation*}
with its mean value
\begin{align*}
    \bar{\Theta}_0&=\frac{1}{L}\int_0^L\frac{1}{{\varphi_0}^2(s)}\mathrm{d}s\\
    &=\frac{1}{\pi}\int_0^\pi\frac{1}{9(\sin^2x+\frac{1}{81}\cos^2x)}\mathrm{d}s=1.
\end{align*}
We conclude that $\psi_0(t,x)=R_0(x)\exp(i[\Theta_0(x)-\mu t])$ is
a MAW of unperturbed system (\ref{1}), see FIG. \ref{Fig 5.1}. The
averaging theorem implies that the MAW $\psi_0(t,x)$ can be
continuable, i.e., there exists a MAW or QMAW
$\psi_\varepsilon(t,x)$ for system (\ref{1}) such that
\begin{equation*}
    \big||\psi_\varepsilon(t,x)|-|\psi_0(t,x)|\big|=
    \mathcal{O}(\varepsilon),
\end{equation*}
for sufficiently small $\varepsilon$. We can have a similar
analysis for the equilibrium $(\theta_1,\rho_1)=(k\pi,10.00)$, see
FIG. \ref{Fig 5.2}.

The equilibrium $(\theta_2,\rho_2)=(k\pi+\pi/2,2.00)$ with the
linearized eigenvalue $\lambda_2=\mp6.44\varepsilon i$ also
persist as periodic orbits for system (\ref{3.1}). Since the
equilibrium is not hyperbolic, one cannot depict the spatial
dynamics of the corresponding continuable periodic orbit. This
question is left open for further study.

\section{Discussion and conclusion}\label{Sec
6}\setcounter{equation}{0}

We have given an abstract result to BECs, see Theorem \ref{TH
4.1}. Using averaging method, we determine the location where the
MAWs or QMAWs exist. Comparing the previous work, we do not
restrict our discussion near the origin since the equation we have
dealt with has some singularities. On the other hand, the MAWs or
QMAWs studied in this paper, which have nontrivial phases, form a
foliation. This is a new result in this context.

However, we can not determine the spatial dynamics of some
periodic orbits corresponding to MAWs or QMAWs since the
equilibriums are not hyperbolic. Maybe second or even higher-order
averaging is required, and this question is left for our further
study. In addition, we can not provide any information upon
modulational instability in time $t$ of MAWs or QMAWs. One reason
is that MAWs or QMAWs obtained in this paper are not standing
waves, but the waves with nontrivial phases, which does not allow
us to apply the method developed by Bronski and
Rapti\cite{bronski2005modulational} directly. This problem may be
not easy but a good topic for further study.

From the view point of a physical application, it might be
reasonable to use the averaging principle to replace a
mathematical model by the corresponding averaged system, to use
the averaged system to make a prediction, and then to test the
prediction against the results of a physical experiment. The study
of this paper exactly gives a frame for application of the
averaging method to BECs.

\section*{Acknowledgements}
The work of Qihuai Liu was supported, in a part, by grant No.
11071181 from National Science Foundation of China and by grant
No. 10KJB110009  from Universities Foundation in Jiangsu Province.
The work of Dingbian Qian was supported  by grant No. 10871142
from National Science Foundation of China.

 \linespread{1.2}
   \selectfont

\end{document}